\documentclass[12pt,reqno]{article}

\usepackage[usenames]{color}
\usepackage{amssymb}
\usepackage{graphicx}
\usepackage{amscd}

\usepackage[colorlinks=true,
linkcolor=webgreen,
filecolor=webbrown,
citecolor=webgreen]{hyperref}

\definecolor{webgreen}{rgb}{0,.5,0}
\definecolor{webbrown}{rgb}{.6,0,0}

\usepackage{color}
\usepackage{fullpage}
\usepackage{float}

\usepackage{graphics,amsmath,amssymb}
\usepackage{amsthm}
\usepackage{amsfonts}
\usepackage{latexsym}
\usepackage{epsf}

\theoremstyle{plain}
\newtheorem{theorem}{Theorem}

\theoremstyle{definition}

\theoremstyle{remark}

\title{Least Periods of $k$-Automatic Sequences}

\author{Daniel Go\v{c} and Jeffrey Shallit\\
School of Computer Science \\
University of Waterloo \\
Waterloo, ON  N2L 3G1 \\
Canada\\
\tt{dgoc@cs.uwaterloo.ca}, \tt{shallit@cs.uwaterloo.ca} 
}

\begin{document}

\maketitle

\begin{abstract}
Currie and Saari initiated the study of
least periods of infinite words, and they showed that
every integer $n \geq 1$ is a least period of the
Thue-Morse sequence.  We generalize this
result to show that the characteristic sequence
of least periods of a $k$-automatic
sequence is (effectively) $k$-automatic.
Through an implementation of our construction, we
confirm the result of Currie and Saari, and we
obtain similar results for the period-doubling sequence,
the Rudin-Shapiro sequence, and the paperfolding
sequence.
\end{abstract}

\section{Introduction}

In a recent paper, Currie and Saari \cite{Currie&Saari:2009} initiated
the study of the least periods of infinite words.  If $x = a_1 \cdots a_n$
is a finite word, then we say $x$ has period $p \geq 1$ if 
$a_i = a_{i+p}$ for $1 \leq i \leq n-p$.  For example,
{\tt alfalfa} has period $3$ and {\tt entanglement} has period $9$.
If ${\bf x} = a_0 a_1 a_2 \cdots $ is an infinite word, then a {\it factor} of
$\bf x$ is a contiguous subword of the form $a_i a_{i+1} \cdots a_j$
for some $i, j$ with $0 \leq i \leq j+1$; we write it as
${\bf x}[i..j]$.  (If $i = j+1$ then the factor
is $\epsilon$, the empty string.)    

Currie and Saari were interested
in the set of all positive integers that can be the least period of some
finite nonempty factor of $\bf x$.  They explicitly computed the set
of least periods
for some famous infinite words, such as the Thue-Morse sequence.

The Thue-Morse sequence ${\bf t} = 0110100110010110 \cdots$ is defined by
letting ${\bf t}[n]$ be the sum of the bits in the binary expansion of
$n$, taken modulo $2$.  They proved that every positive integer can be
the least period of the Thue-Morse sequence.

The Thue-Morse sequence is one of a much larger class of infinite words
called ``automatic''.  Roughly speaking, an infinite word ${\bf x}$ is
$k$-automatic if there is a deterministic finite automaton taking the
base-$k$ representation of $n$ as input, with transitions leading to 
a state with output ${\bf x}[n]$.  For more details, see
\cite{Cobham:1972,Allouche&Shallit:2003}.

In this note, we prove that if $\bf x$ is $k$-automatic, then so is the
characteristic sequence of the least periods of $\bf x$.   Our method
gives an explicit way to construct the automaton accepting the
base-$k$ representation of the least periods of $\bf x$.
Using an implementation developed by the first author,
we then 
reprove the Currie-Saari result for Thue-Morse using a short
finite computation, and we find similar results for three other
classic sequences.

\section{The main result}

\begin{theorem}
If $\bf x$ is a $k$-automatic sequence, then the characteristic
sequence of least periods of $\bf x$ is (effectively) $k$-automatic.
\end{theorem}

\begin{proof}
Using the method developed in \cite{Allouche&Rampersad&Shallit:2009,Charlier&Rampersad&Shallit:2011}, it suffices to construct a predicate
$L(n)$ that is true if $n$ is a least period and false otherwise,
using a logical language restricted to addition, subtraction, indexing
into $\bf x$, comparisons, logical operations, and the existential and
universal quantifiers.

It is easy to express the predicate $P$ that $n$ is a period of
the factor ${\bf x}[i..j]$, as follows:
\begin{eqnarray*}
P(n,i,j) &=& {\bf x}[i..j-n] = {\bf x}[i+n..j]  \\
&=& \forall \ t \text{ with $i \leq t \leq j-n$ we have } 
	{\bf x}[t] = {\bf x}[t+n] .
\end{eqnarray*}
Using this, we can express the predicate $LP$ that $n$ is the least
period of ${\bf x}[i..j]$:
$$ LP(n,i,j) = P(n,i,j) \text{ and } \forall n' < n \  \neg P(n',i,j).$$
Finally, we can express the predicate that $n$ is a least period
as follows
$$L(n) = \exists i, j \geq 0 \text{ with $0 \leq i+n \leq j-1$ }
	LP(n, i, j) .$$
The construction is effective, and there is an algorithm
that, given the automaton generating $\bf x$, will
produce an automaton generating the characteristic
sequence of least periods of $\bf x$.
\end{proof}	

\section{Computations}

Currie and Saari \cite[Thm.\ 2]{Currie&Saari:2009} proved

\begin{theorem}
For each integer $n \geq 1$, the Thue-Morse word has a factor of
period $n$.
\end{theorem}

We implemented the algorithm in
\cite{Goc&Henshall&Shallit:2012} to convert
the automaton generating a $k$-automatic sequence $\bf x$ to the automaton
accepting the characteristic sequence of least periods of $\bf x$.
Using this, we were able to verify the result above using a short
computation.  (In contrast, Currie and Saari used four pages of rather
intricate case reasoning.)

We also carried out the same computation for three other famous infinite
words:
\begin{itemize}
\item The period-doubling sequence ${\bf d} = d_0 d_1 \cdots$ defined
by $d_n = 1$ if ${\bf t}[n] \not= {\bf t}[n+1]$, and $0$ otherwise;
\item The paperfolding sequence $\bf p$, defined as the limit of the finite
words $p_0 = 0$ and $p_{i+1} = p_i \ 0 \ \overline{p_i}^R$, where
$\overline{0} = 1$, $\overline{1} = 0$, and $w^R$ denotes the mirror 
image or reversal of the word $w$;
\item The Rudin-Shapiro sequence ${\bf r} = r_0 r_1 \cdots$ defined
by counting the number of (possibly overlapping) occurrences
of $11$ in the binary representation of $n$, taken modulo $2$.
\end{itemize}

Our results can be summarized as follows:

\begin{theorem}
For each integer $n \geq 1$, the period-doubling sequence and the
Rudin-Shapiro sequence have a factor of least period $n$.

For the paperfolding sequence, the least periods are given by
the integers whose base-$2$ representations are accepted by
the automaton below.  The least omitted least period is $18$,
and there are infinitely many.  In the limit, exactly
$57/64$ of all integers are least periods of the paperfolding
sequence.

\begin{figure}[H]
\leavevmode
\def\epsfsize#1#2{1.8#1}
\centerline{\epsfbox{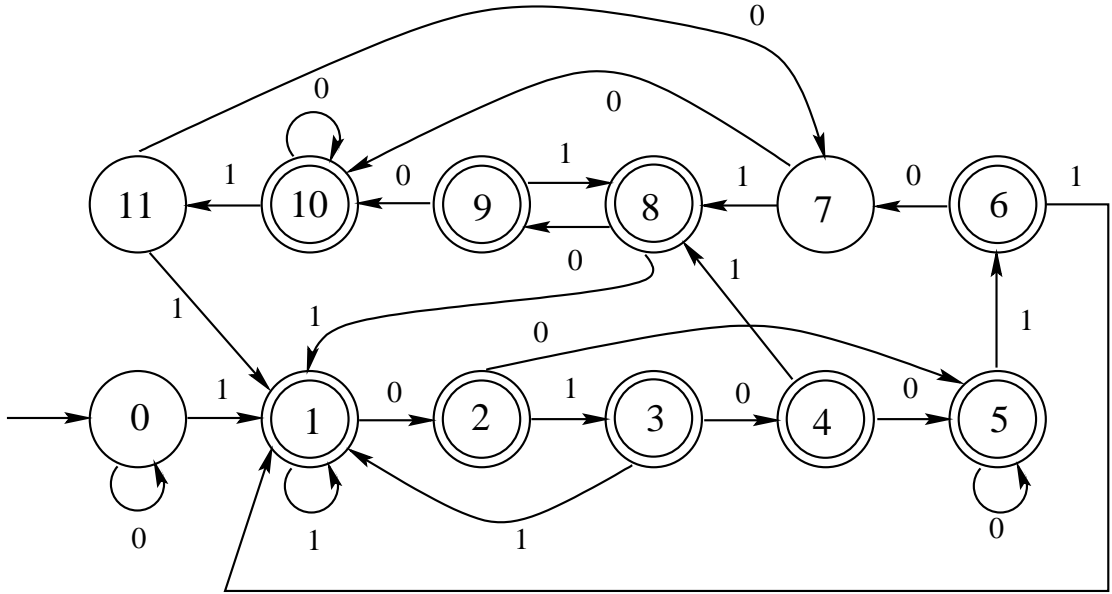}}
\protect\label{fig1}
\caption{A finite automaton accepting least periods of the paperfolding
sequence}
\end{figure}
\end{theorem}

\begin{proof}
The first results were obtained through our algorithm.  A summary
of our computations appears below:
\begin{table}[H]
\begin{center}
\begin{tabular}{cccc}
\hline
Sequence & Number of states in largest & Number of states in & Seconds of \\
name & intermediate automaton & final automaton & CPU time \\
\hline
Thue-Morse & 264 & 1 & 5.882 \\
Rudin-Shapiro & 1029 & 1 & 27.797 \\
Period-doubling & 89 & 1 & 4.327 \\
Paper-folding & 393 & 12 & 11.597\\
\hline
\end{tabular}
\end{center}
\end{table}

For the
result about the paperfolding sequence, we take the automaton
computed by the algorithm (displayed in Figure 1) and compute
the transition matrices $M_a$, $a \in \lbrace 0,1\rbrace$,
containing a $1$
in row $i$ and column $j$ if there is a transition on $a$ from
state $i$ to state $j$.   Then $(M^n)_{i,j}$, where
$M := M_0 + M_1$, gives the number of words taking the automaton
from state $i$ to state $j$.  A short computation gives
that each row of 
$\lim_{n \rightarrow \infty} 2^{-n} M^n $
equals 
$${1 \over 64}[0,16,8,4,2,10,5,4,4,2,6,3].$$
All states except $7$ and $11$ are accepting,
so the density of least periods is given by
$(64-4-3)/64 = 57/64$, as claimed.
\end{proof}

\end{document}